\documentclass[twocolumn,aps,superscriptaddress,amsmath,amssymb]{revtex4-2}

\usepackage{epsfig}
\usepackage{graphicx}
\usepackage{amsmath}
\usepackage{amssymb}
\usepackage{dcolumn}
\usepackage{bm}
\usepackage{xcolor}
\usepackage{physics}
\usepackage{dsfont}
\usepackage{mathtools}
\usepackage{subfigure}
\usepackage{float}
\usepackage{comment} 
\usepackage{amsthm}
\usepackage{makecell}
\usepackage{hyperref}
\usepackage[normalem]{ulem}

\newcommand{\be} {\begin{equation}} 
\newcommand{\ee} {\end{equation}}
\DeclarePairedDelimiter\ceil{\lceil}{\rceil}
\DeclarePairedDelimiter\floor{\lfloor}{\rfloor}

\newtheorem{theorem}{Theorem}

\newtheorem{corollary}{Corollary}
\newtheorem{definition}{Definition}

\newtheoremstyle{noparentheses}
  {\topsep}   
  {\topsep}   
  {\itshape}  
  {0pt}       
  {\bfseries} 
  {.}         
  {5pt plus 1pt minus 1pt} 
  {\thmname{#1} \thmnumber{#2} \thmnote{#3}}          
\theoremstyle{noparentheses}
\newtheorem{definition*}[definition]{Definition}

\newtheoremstyle{noparentheses}
  {\topsep}   
  {\topsep}   
  {\itshape}  
  {0pt}       
  {\bfseries} 
  {.}         
  {5pt plus 1pt minus 1pt} 
  {\thmname{#1} \thmnumber{#2} \thmnote{#3}}          
\theoremstyle{noparentheses}
\newtheorem{fact*}[fact]{Fact}

\newtheoremstyle{noparentheses}
  {\topsep}   
  {\topsep}   
  {\itshape}  
  {0pt}       
  {\bfseries} 
  {.}         
  {5pt plus 1pt minus 1pt} 
  {\thmname{#1} \thmnumber{#2} \thmnote{#3}}          
\theoremstyle{noparentheses}
\newtheorem{protocol*}[protocol]{Protocol}

\begin{document}

\title{Geometrical constructions of purity testing protocols and their applications to quantum communication}

\author{R\'obert Tr\'enyi}
\address{Department of Theoretical Physics,  University of the Basque Country UPV/EHU, P.O. Box 644, E-48080 Bilbao, Spain}
\address{EHU Quantum Center, University of the Basque Country UPV/EHU, Barrio Sarriena s/n, E-48940 Leioa, Biscay, Spain}
\address{HUN-REN Wigner Research Centre for Physics, P.O. Box 49, H-1525 Budapest, Hungary}
\address{Department of Theoretical Physics, University of Szeged, Tisza L. krt. 84-86, H-6720 Szeged, Hungary}

\author{Simeon Ball} \address{Departament de
Matem\`atiques, Universitat Polit\`ecnica de Catalunya, Jordi Girona 1-3 08034
Barcelona, Spain}

\author{David G. Glynn}
\address{College of Science and Engineering, Flinders University, G.P.O. Box 2100, SA 5001, Australia}

\author{Marcos Curty}
\address{Vigo Quantum Communication Center, University of Vigo, Vigo E-36310, Spain}
\address{Escuela de Ingenier\'ia de Telecomunicac\'ion, Department of Signal Theory and Communications, University of Vigo, Vigo E-36310, Spain}
\address{atlanTTic Research Center, University of Vigo, Vigo E-36310, Spain}

\begin{abstract}
Purity testing protocols (PTPs), i.e., protocols that decide with high
probability whether or not a distributed bipartite quantum state is maximally
entangled, have been proven to be a useful tool in many quantum communication
applications. In this paper, we provide geometrical constructions for such
protocols that originate directly from classical linear error correcting codes
(LECCs), in a way that the properties of the resulting PTPs are completely
determined from those of the LECCs used in the construction. We investigate the
implications of our results in various tasks, including error detection,
entanglement purification for general quantum error models and quantum message
authentication.
\end{abstract}

\maketitle

\section{Introduction}

Entangled states~\cite{HorodeckiEntanglement,toth_ent,Friis_ent_rev}, i.e.,
states that exhibit quantum correlations, are a crucial resource in many
applications of quantum information science, including quantum communication,
quantum computing and quantum metrology. In quantum communication, it has been
proven that a necessary precondition for successful quantum key distribution
(QKD), is that the legitimate users of the system (Alice and Bob) can detect the
presence of entanglement in a quantum state that is effectively distributed
between them~\cite{curtyPRL,preconditionAcin}. Entanglement also allows the
reliable transmission of quantum information over noisy and lossy channels,
since it can be employed to obtain a perfect quantum channel. Once Alice and Bob
share perfect Einstein-Podolsky-Rosen~\cite{EPR} (EPR) pairs, they can use them
to teleport~\cite{TeL} any quantum state with the aid of
classical communication. This procedure represents an alternative solution to
that based on quantum error correction
codes~\cite{Shor95,Steane96,ECOverGF4,GottesmanIntro}, where errors are actively
corrected after the transmission of the state. 

Due to the central status of entanglement in many quantum communication
scenarios, including the future quantum internet~\cite{quantum_internet}, a very
significant amount of research has been dedicated to the problem of finding good
criteria for separability~\cite{Separability,HorodeckiEntanglement,toth_ent}.
While the complete solution to this question is an NP-HARD
problem~\cite{Gurvits}, one can nevertheless find hierarchies of sufficient
criteria for entanglement that involve solving efficiently a convex
optimization problem in each step of the
hierarchy~\cite{Doherty,Doherty2,Doherty3,Doherty4}. Entanglement witnesses also
provide a way to decide whether a state is entangled~\cite{toth_ent}.

The manipulation of entangled states in noisy environments has also received
great attention. Entanglement purification protocols~\cite{Ben96a,Ben96,Ben96b,
HorodeckiEntanglement,distLuo} can distill perfect EPR pairs from a larger
number of noisy entangled states. The case where the initial states are
identical copies of a particular pure two-qubit entangled state was studied
in~\cite{Ben96a}. This result was then extended to the mixed states
scenario in~\cite{Ben96,Ben96b,Horodecki97} and also distillation experiments have
been performed~\cite{dist1,dist2}.

Although the ability to correct quantum errors is an essential ingredient in
many quantum communication protocols, there are also situations where it is
enough to detect with high probability when an error has
occurred~\cite{BARNUM_2002,Gisin04,err_det_recent}. If an error is detected, the
protocol simply discards the signal, i.e., the error is transformed into an
erasure. Such method has the potential of being simpler to implement than
correcting errors. 


A common starting point in the design of entanglement purification (error
 detection) protocols is a model for the source of errors to be corrected
 (detected). Relatively simple error models are often assumed, such as to
 consider that Alice and Bob share identical copies of the same state or,
 equivalently, that the noise acts independently on each signal. Indeed, this
 description is justified in many communication scenarios from technological
 considerations. However, there are also situations, like in the context of most
 cryptographic protocols, where the action of the channel is controlled by a
 third party (Eve), and thus the assumption of independent errors is not valid
 anymore~\cite{Lo99,Shor00,BARNUM_2002,Lo_Secure,Xu_Review}. Notably,
 entanglement purification and error detection schemes can also be adapted to
 work outside the independent error model \cite{BARNUM_2002,Ambainis02}. An
 interesting tool to achieve this goal is the use of purity testing protocols
 (PTPs). Basically, a PTP is an error detection scheme that can distinguish the
 state of perfect EPR pairs from any other state. These protocols have been used
 implicitly by Lo and Chau \cite{Lo99}, and Shor and Preskill~\cite{Shor00} in
 the context of security proofs for QKD. These results prove that it is possible
 to determine with very high accuracy whether or not a quantum state is a tensor
 product of EPR pairs. Remarkably, in the context of quantum message
 authentication, Barnum {\it et al.} \cite{BARNUM_2002} showed explicitly how to
 construct PTPs from purity testing codes (PTCs). The latter are sets of quantum
 error correcting codes (QECCs) that satisfy that most of the codes in the
 family detect any particular Pauli error (see Definition~\ref{def:Pauli_error}
 below). Moreover, by using results from projective geometry,~\cite{BARNUM_2002}
 demonstrated how to obtain a PTC with this covering property. Subsequent to
 their work, Ambainis, Smith and Yang~\cite{Ambainis02} pointed out that PTPs
 can also be used for entanglement purification even when no information about
 the error source is available but only about the fidelity of the shared state.
 In fact, PTPs can be considered as a special case of entanglement
 certification~\cite{Friis_ent_rev}.

In this paper we investigate PTPs in the same spirit as Barnum {\it et
al.}~\cite{BARNUM_2002}, and we generalize their results to show that they can
be constructed (see Theorem~\ref{Maintheorem} below) from classical linear error
correcting codes (LECCs). The analysis is based on known results in projective
geometry, but in contrast to~\cite{BARNUM_2002}, we remove the need of
considering the so-called normal rational curves in finite projective
spaces~\cite{Beut98,Hirs98,Casse_book}. Instead, we show that the
construction in~\cite{BARNUM_2002} corresponds to a particular LECC that
satisfies the Singleton bound~\cite{Mac77}. 

We explore the implications of our results for error detection~\cite{Gisin04},
entanglement purification and quantum authentication~\cite{BARNUM_2002,Ambainis02}. We note that PTCs have also been used in
secret sharing and secure multiparty computation~\cite{MultiPartyComputation,
ErrCorr, cipherQuantum} and our results might be relevant there as well. In the
case of quantum message authentication, since the
secret key required is an expensive resource,
schemes with key recycling have also been
considered~\cite{Portmann,HLM,reuse,cipherQuantum,garg,Fehr}. In this regard, we
show that our method allows reusing most of the key in subsequent rounds of the
authentication protocol. That is, our construction actually gives a so-called
strong stabilizer PTC (SPTC), where SPTC means that the PTC is constructed
from stabilizer QECCs.

The paper is organized as follows. In Section~\ref{PTP} we define formally what
is a PTP and we show how to obtain such a protocol from PTCs composed of
stabilizer QECCs. Section~\ref{Geo_without_barnum} contains the main results of
the paper. There, we present geometrical constructions for PTPs and PTCs based
on results from projective geometry. In particular, we show how PTPs and PTCs
can be obtained from classical LECCs. Section~\ref{Appl} analyses some
applications of the proposed PTPs. This includes error detection, entanglement
purification schemes for general quantum error models and quantum authentication
protocols. More specifically, we evaluate the performance of PTPs coming from
two families of LECCs. Finally, Section~\ref{Concl} summarizes our findings. The
paper has three Appendices, where we provide concrete examples for some abstract
mathematical notions that we use. In Appendix~\ref{app:companion_matrix}, we
describe the companion matrix formalism to represent the elements of finite
fields as matrices. In Appendix~\ref{app:symplectic_form}, we introduce how to
describe elements of the $n$-qubit Pauli group as binary strings. Finally, in
Appendix~\ref{app:example_stabilizers}, we provide a specific example that
elucidates the stabilizers constituting the PTCs. 

\section{Purity testing Protocols}\label{PTP}

A PTP is a quantum operation that can be implemented via local operations and
classical communication (LOCC), allowing Alice and Bob to check with high
confidence whether or not the quantum state they share corresponds to $n$
copies of the EPR pair $|\Phi^+\rangle=1/\sqrt{2}(|00\rangle+|11\rangle)$. When
the answer is negative the protocol discards the state. This means that
some potential EPR pairs might be sacrificed in the test process.

\begin{definition*}[\cite{BARNUM_2002}]\label{Def:PTP} A PTP with error
$\epsilon$ is a LOCC quantum operation ${\cal O}$ which maps $2n$ qubits (half
held by Alice and half held by Bob) to $2m+1$ qubits ($m+1$ held by Alice
and $m$ held by Bob) and satisfying the following two conditions:
\begin{enumerate}
\item {\textrm Completeness:} ${\cal
O}(|\Phi^{+}\rangle^{\otimes{}n})=|\Phi^{+}\rangle^{\otimes{}m}\otimes|{\rm
ACC}\rangle$,
\item {\textrm Soundness:} $\Tr(P\ {\cal O}(\rho))\geq{}1-\epsilon\quad\forall\
\rho$,
\end{enumerate}

\noindent where $P$ represents the projection on the subspace spanned by
$|\Phi^{+}\rangle^{\otimes{}m}\otimes|{\rm ACC}\rangle$ and
$|\psi\rangle\otimes{}|{\rm REJ}\rangle$, $\forall\ |\psi\rangle$. The states
$\ket{{\rm ACC}}$ and $\ket{{\rm REJ}}$ are orthogonal single qubit states,
representing whether the parties accept or reject the input state as
$\ket{\Phi^{+}}^{\otimes n}$, respectively.
\end{definition*}

\noindent We emphasize that since 
\begin{equation}P=\mathds{1}_{2m+1}-\left(\mathds{1}_{2m}-\ketbra{\Phi^{+}}{\Phi^{+}}^{\otimes{}m}\right)\otimes \ketbra{{\rm ACC}}{{\rm ACC}},
\end{equation} where  $\mathds{1}_{2m+1}$ denotes the identity operator on the $2m+1$ qubit space, the soundness condition can be written as
\begin{align}&\textrm{Tr}(P\ {\cal O}(\rho))=\\ \nonumber &1-\textrm{Tr}\left\lbrace\left[\left(\mathds{1}_{2m}-\ketbra{\Phi^{+}}{\Phi^{+}}^{\otimes{}m}\right)\otimes \ketbra{{\rm ACC}}{{\rm ACC}}\right] {\cal O}(\rho)\right\rbrace.
\end{align}
From this it is clear that Definition~\ref{Def:PTP} requires that the
probability of accepting a quantum state different from
$\ket{\Phi^{+}}^{\otimes{}m}$ is smaller or equal than $\epsilon$. 

Basically, a PTP can be interpreted as a protocol that approximates the von
Neumann measurement given by the projection onto
$|\Phi^{+}\rangle^{\otimes{}m}$ and its orthogonal complement.

Before providing a method for constructing PTPs,
 let us define the notions of a Pauli error and the Pauli group.
\begin{definition}\label{def:Pauli_error} An $n$-qubit Pauli error, $E_t$,
   is a unitary operator of the form $E_t=c\,w_1\otimes\ldots\otimes{}w_n$,
   where each $w_j$ represents a Pauli matrix ($I, \sigma_x, \sigma_y,
   \sigma_z$) and the phase factor $c\in \{1,-1,i,-i\}$. The set of all Pauli
   errors is called the Pauli group, denoted as $E$, and it is a subgroup of the
   unitary group $U(2^n)$. \end{definition} 
   
 A particularly
efficient method to construct PTPs is the one proposed by Barnum {\it et al.}
\cite{BARNUM_2002}. It is based on the use of special sets of stabilizer
QECCs \cite{Got97}, $\{Q_k\}$, with the following property: given an
arbitrary non-trivial Pauli error $E_t$, if we select at random and {\it a
posteriori} a $Q_k$ within the set, then the probability that $Q_k$ does not
detect the error $E_t$ is bounded by a parameter $\gamma$. Such a set of
QECCs is called a stabilizer purity testing code (SPTC).

\begin{definition*}[\cite{BARNUM_2002,Portmann}]\label{PTC} An SPTC with error
probability $\gamma$ is a set of stabilizer QECCs $\{Q_k\}$, with $k\in{\cal
K}$, such that for all Pauli errors $E_t$ in the Pauli group $E$, with
$E_t\neq{}\mathds{1}$, and for $k$ selected at random in ${\cal K}$ then

\begin{equation}
Pr(E_t\in Q_k^{\perp}\setminus Q_k)\leq\gamma.
\end{equation}
In other words, the percentage of $k$'s in ${\cal K}$ correcting $E_t$ is
at most $\gamma$. Moreover, if
\begin{equation}
  Pr(E_t\in Q_k^{\perp})\leq\gamma, \end{equation}then the set is
  a strong SPTC with error probability $\gamma$.
\end{definition*}
We remark that with this notation the stabilizer $Q_k$ is an
abelian subgroup of the Pauli group $E$ and $Q_k^{\perp}$ is the centralizer of
this subgroup $Q_k$ in $E$~\cite{Nielsen}, meaning that it contains the Pauli
errors that commute with all the elements of $Q_k$ (errors in $Q_k$ are not
detectable by the quantum code).

Given an SPTC, $\{Q_k\}$, with error $\gamma$ and a quantum state $\rho$, it is
straightforward to construct a PTP of error $\epsilon=\gamma$ to test $\rho$,
therefore from now on we will also use $\epsilon$ to denote the error of the
SPTCs. For this Alice and Bob need to agree first on a particular random
$k\in{\cal K}$. Subsequently, they need to measure the syndrome of $Q_k$ in
their respective quantum subsystems. If Alice and Bob actually hold the state
$\ket{\Phi^{+}}^{\otimes{}n}$ then upon their measurements they will obtain the
same strings due to the quantum correlations present in the
$\ket{\Phi^{+}}^{\otimes{}n}$ state. But if they hold an erroneous state, say,
$(E_t\otimes\mathds{1}_n)\ket{\Phi^{+}}^{\otimes{}n}$ with $E_t\neq{}\mathds{1}$
being a Pauli error (see Definition~\ref{def:Pauli_error}), then it is likely
that they will find different syndromes. Note that due to the freedom in
choosing the encoded logical basis states~\cite{Nielsen} in the stabilizer space
of $Q_k$ we can think of each half of the $\ket{\Phi^{+}}^{\otimes{}n}$ state as
if they were the encoded versions of the corresponding halves of the
$\ket{\Phi^{+}}^{\otimes{}m}$ state ($m\leq n$) with a specific syndrome. This
means that if both syndromes are the same Alice and Bob accept the quantum
state, perform the decoding procedure for the code $Q_k$ and obtain a
quantum state close to $\ket{\Phi^{+}}^{\otimes{}m}$ with $m\leq n$; otherwise
they discard the quantum state. For a precise proof of the above statement we
refer to~\cite{BARNUM_2002}. In this way the problem of constructing PTPs can be
reduced to the problem of obtaining SPTCs.

Next we present an efficient method to create SPTCs (and, therefore, also PTPs)
from classical LECCs. 

\section{Geometrical construction}\label{Geo_without_barnum} 

The method for constructing SPTCs in~\cite{BARNUM_2002} constitutes a
special case of a more general principle based on the use of classical linear
error correcting codes. In this Section, we prove in Theorem~\ref{thm:ptc}, that
any classical LECC $[c,2r,d]_q$ over the finite (Galois)
field $GF(q)$ naturally gives a SPTC with parameters that follow directly from
those of the LECC. This is possible because every column
of the generator matrix of the code $[c,2r,d]_q$ corresponds to a point in
$PG(2r-1,q)$ and, as we prove in Theorem~\ref{thm:points_stab}, every point in
the projective space $PG(2r-1,q)$ gives a stabilizer QECC. Most importantly, the
general properties of a LECC guarantee that the set of
stabilizers obtained in that way form a SPTC. The advantage of this method is
that it significantly simplifies the process of obtaining SPTCs when compared
to that based on the use of normal rational curves in finite projective
spaces~\cite{BARNUM_2002}. 

We refer the reader to~\cite{Beut98,Hirs98,Casse_book,mullen_fields,lidl_finite_fields} for a more extensive
list of properties regarding the finite field $GF(q)$ and the projective space
$PG(2r-1,q),$ over $GF(q).$ For simplicity, we shall assume that $q=2^s$
throughout the paper. We remark, however, that the proofs below can be
generalized straightforwardly for the case when $q=p^s,$ where $p>2$ is a
prime, where qudits instead of qubits are involved.

The projective space $PG(2r-1,q)$ is basically the lattice of all subspaces of
the vector space $V(2r,q).$ The homogeneous coordinates $(a_0, ...,
a_{2r-1})$ for a point in $PG(2r-1,q)$ with $a_i\in GF(q)$ such that not all
$a_i$'s are zero form a vector generating the corresponding $1$-dimensional
subspace of $V(2r,q)$. This means that $\lambda(a_0,
..., a_{2r-1})$ represents the same point in $PG(2r-1,q)$ for all $\lambda\in
GF(q)\setminus\{0\}$. 

That is, the coordinates of a point of $PG(2r-1,q)$ are elements of
$GF(2^s).$ Precisely, one can think of $GF(2^s)$ as an $s$-dimensional vector
space $V(s,2)$ over $GF(2),$ where $GF(2)=\{0,1\}$ with the usual multiplication
and addition modulo 2~\cite{mullen_fields,Casse_book}. This means, therefore,
that one can equivalently think of a point in $PG(2r-1,q)$ as a
1-dimensional subspace of $V(2r,q)$ or as an $s$-dimensional
subspace of $V(2rs,2)$. Another useful way to represent the
elements of $GF(2^s)$ is via, for example, the companion matrix
formalism~\cite{lidl_finite_fields,Albert37}, where each element of $GF(2^s)$ corresponds
to an $s\times s$ matrix with entries from $GF(2)$. We describe this latter
method in Appendix~\ref{app:companion_matrix}, where for illustration purposes
we provide a specific example for representing the field $GF(4)$. This shows
that a point in $PG(2r-1,2^s)$ defines an $s$-dimensional subspace of
$V(2rs,2).$ The generators of this subspace are vectors in $V(2rs,2)$ and each
of these $s$ vectors defines a Pauli error (i.e., they are elements of the
$n$ qubit Pauli group but without phase factors) via the correspondence
described in Appendix~\ref{app:symplectic_form}~\cite{Ball_geo,GottesmanIntro}.
In this latter Appendix, we also introduce a canonical symplectic form over $V(2rs,2)$
that captures the commutation relation between Pauli errors. These Pauli errors
are the generators of a stabilizer QECC as we prove in
Theorem~\ref{thm:points_stab}.  

As it is described in Appendix~\ref{app:symplectic_form}, the subspace that the
Pauli errors (i.e., the vectors in $V(2rs,2)$) form, has to be totally isotropic
with respect to a non-degenerate symplectic form to obtain a stabilizer QECC.
Isotropy makes sure that the generators of the stabilizer commute.

It is clear how the canonical symplectic form from
Appendix~\ref{app:symplectic_form} captures whether the corresponding Pauli
errors (of the vectors in $V(2rs,2)$) commute. In the proof, however, for the
sake of conciseness, we use a symplectic form~\cite{BARNUM_2002}
based on the field trace~\cite{mullen_fields}:
\begin{equation}\label{eq:trace_form}
  (x,y)_{\Tr}=\Tr_{2^{2rs}\rightarrow 2}\left(xy^{2^{rs}}\right),
\end{equation}
where $x,y\in GF(2^{2rs})\equiv V(2rs,2)$ represent elements of the $rs$-qubit
Pauli group via the correspondence described in
Appendix~\ref{app:symplectic_form} and the field trace is defined as
\begin{equation}\label{eq:field_trace}
\Tr_{b^s\rightarrow b}(x)=x+x^b+...+x^{b^{s-1}},
\end{equation} 
where $b$ is a prime power and $x\in GF(b^s).$ Note that we can use the form
from Eq.~\eqref{eq:trace_form} due to the fact that all non-degenerate
symplectic forms are equivalent on $V(2r,q)$ and also on $V(2rs,2)$. Therefore,
the field trace-based symplectic form captures the same commutation relations
between the generators as the canonical symplectic form from
Appendix~\ref{app:symplectic_form}. Considering all the above facts we can now
prove the following Theorem.

\begin{theorem}\label{thm:points_stab} A point in the projective space
$PG(2r-1,2^s)$ corresponds to a stabilizer quantum error correcting
code $[[rs,rs-s]]$ encoding $rs-s$ qubits into $rs$ qubits.  
  \end{theorem} 
  \begin{proof}
    As above, we consider a point in $PG(2r-1,2^s)$ as $\lambda x\neq 0$, where $x \in
    GF(2^{2rs})$ and $\lambda \in GF(2^s)$. To prove that this $s$-dimensional
    subspace is totally isotropic with respect to $(x,y)_{\rm Tr},$ defined in
    Eq.~\eqref{eq:trace_form}, we need to prove that $(\lambda x,\mu x)_{\rm
    Tr}=0$, for all $\lambda$, $\mu \in GF(2^s)$. 
  Since $\mu \in GF(2^s)$ we have that 
  \begin{align}(\lambda x,\mu x)_{\rm
  Tr}&=\Tr_{2^{2rs}\rightarrow 2} (\lambda \mu^{2^{rs}} x^{2^{rs}+1})
  \\\nonumber &=\Tr_{2^{2rs}\rightarrow 2} (\lambda \mu x^{2^{rs}+1}).\end{align} 
  Since $$\Tr_{2^{2rs}\rightarrow 2}
  (y)=\Tr_{2^s\rightarrow 2}\left(\Tr_{2^{2rs}\rightarrow 2^s}(y)\right),$$
  for all $y\in GF(2^{2rs}),$ we can write that 
  \begin{align}
  (\lambda x,\mu x)_{\Tr}=\Tr_{2^s\rightarrow 2}\left(\lambda \mu\Tr_{2^{2rs}\rightarrow 2^s}\left(x^{2^{rs}+1}\right)\right),
  \end{align}
  which is zero since 
  \begin{align}
    &\Tr_{2^{2rs}\rightarrow 2^s}(x^{2^{rs}+1})=\big[x^{2^{rs}+1}+(x^{2^{rs}+1})^{2^s}+\cdots \\\nonumber
    &+(x^{2^{rs}+1})^{2^{s(r-1)}}\big] +\big[(x^{2^{rs}+1})^{2^{sr}}+\cdots +(x^{2^{rs}+1})^{2^{s(2r-1)}}\big]\\\nonumber
    &=\Tr_{2^{rs}\rightarrow 2^s}(x^{2^{rs}+1})+\Tr_{2^{rs}\rightarrow 2^s}(x^{2^{rs}+1})=0.
  \end{align}
  Note that here we use the relation $$(x^{2^{rs}+1})^{2^{rs}}=x^{2^{rs}+1},$$ 
  and the fact that the characteristic of the underlying field is 2, which implies that
  $1=-1.$ In short, with this, we have proven that the subspace corresponding to a point in
  $PG(2r-1,2^s)$ is totally isotropic with respect to a symplectic form,
  therefore it gives a stabilizer QECC. The stabilizer space has $s$ independent
  generators and each generator corresponds to $rs$ qubits as it is shown in
  Appendix~\ref{app:symplectic_form}. So this means that it encodes $rs-s$
  qubits into $rs$ qubits and is an $[[rs,rs-s]]$ QECC.
  \end{proof}

For an explicit construction of the stabilizers, we refer the reader to
Appendix~\ref{app:example_stabilizers}, where we use the companion
matrix formalism from Appendix~\ref{app:companion_matrix} to obtain the
$s$-dimensional subspace of $V(2rs,2).$ Then we employ the correspondence
introduced in Appendix~\ref{app:symplectic_form} to obtain the Pauli errors
constituting each stabilizer QECC.

\begin{theorem}\label{thm:ptc} Given a $[c,2r,d]_q$ LECC,
$C$, where $q=2^s$, then there is a strong SPTC consisting of $c$ stabilizer QECCs,
$\{Q_k\}$, with parameters $[[rs,rs-s]]$ such as its error rate is
$\epsilon=1-d/c$.
\end{theorem}

\begin{proof}
A generator matrix $G$ of the code $C$ is a $2r\times c$ matrix with entries
from $GF(2^s).$ Thus, each column of the matrix $G$ corresponds to a point in
$PG(2r-1,2^s).$ This matrix has $c$ columns, so based on
Theorem~\ref{thm:points_stab}, this means that we have $c$ stabilizer QECCs with
each of them encoding $rs-s$ qubits into $rs$ qubits. Next, we have to determine
that if we take a random Pauli error from the $rs$-qubit Pauli group then at
most how many of these $c$ stabilizer QECCs (i.e., abelian subgroups)
can fail in detecting the error. An error is undetectable for a stabilizer if it
commutes with all the generators of the code (i.e., if it is in the centralizer
of the stabilizer group). We remark, that it is possible that such an ``error''
is actually in the stabilizer subspace, in which case it is not a true error
since it does not alter the states stabilized by the code, nevertheless we
consider it as an undetectable error. 

Let us now formalize mathematically what is an undetectable error for a
stabilizer. Recall that we can consider the vectors of $V(2r,2^s)$ as elements
of $GF(2^{2rs})$, so a Pauli error $e$ is given by such an element (i.e., this
is a representation of a certain $E_t$ via for example the companion matrix
representation). It commutes with the abelian subgroup we obtain from the column
$x$ if and only if $(\lambda x,e)_{\Tr}=0$ for all $\lambda \in GF(2^s)$, i.e., the
symplectic product of $e$ and $\lambda x$ is zero. Since all non-degenerate symplectic forms are equivalent,
considering $e$ and $\lambda x$ now in the vector space $V(2r,2^s)$ (which is
equivalent to considering it as a vector in $V(2rs,2)$), there is an equivalent
symplectic form to $(x,e)_{\Tr}$ given by
\begin{equation}\label{eq:symplectic_product}
(x,e)_{\tr}=\sum_{j=1}^r \Tr_{2^s \rightarrow 2}(e_{r+j}x_j-e_j x_{r+j}),
\end{equation} 
where we note that the minus
sign is not necessary since we are working modulo 2. Now, as mentioned
before, $e$ commutes with $x$ iff
\begin{equation}
0=(\lambda x,e)_{\tr}=\sum_{j=1}^r \Tr_{2^s \rightarrow 2}\big( \lambda (e_{r+j}x_j-e_jx_{r+j})\big)
\end{equation}
for all $\lambda \in GF(2^s)$. We want to conclude that
\begin{equation}\mu\coloneq(x,e)=\sum_{j=1}^r (e_{r+j}x_j-e_j x_{r+j})=0,
\end{equation}
where $\mu \in GF(2^s).$ By Theorem 4 from~\cite{Seroussi} there is a
basis $B=\{b_1,\ldots,b_s\}$ for $GF(2^s)$ over $GF(2)$ with the property that
$\Tr_{2^s \rightarrow 2}(b_ib_j)=\delta_{ij}$, where $\delta_{ij}$ is the
Kronecker delta. Considering $\lambda$ and $\mu$ over the basis $B$ we can write
$$
\lambda=\sum_{k=1}^s \lambda_k b_k, \ \, \mu=\sum_{l=1}^s \mu_l b_l.
$$
Then
\begin{align}\label{eq_symplectic_form}
&\sum_{j=1}^r\Tr_{2^s \rightarrow 2}\big( \lambda (e_{r+j}x_j-e_jx_{r+j})\big)\\\nonumber
&=\sum_{k,l} \Tr_{2^s \rightarrow 2}(\lambda_k\mu_l b_kb_l)\\\nonumber
&=\sum_{k,l} \lambda_k\mu_l \Tr_{2^s \rightarrow 2}(b_kb_l)\\\nonumber
&=\lambda_k\mu_l \delta_{kl}
\end{align}
Since this is zero for all $\lambda \in GF(2^s)$, we can choose $\lambda_k=1$
and $\lambda_l=0$ for $l\neq k$ and conclude that $\mu_k=0$ for all $k \in
\{1,\ldots,s\}$. Hence $\mu=(x,e)=0$. Thus, considering a Pauli error as
\begin{equation}e=(e_{r+1},\dots,e_{2r},-e_1,\dots,-e_r),\end{equation} we see that the zero coordinates
of the codeword
\begin{equation}\label{eq:error}
eG=e'=(e'_1,e'_2,...,e'_{c})
\end{equation} 
indicate the abelian subgroups with which $e$ commutes. Thus, if
$e$ commutes with the $i$-th stabiliser, meaning that it is an undetectable
error for the $i$-th stabilisier, then $e_i'=0$. The codeword $eG$ has at most
$c-d$ zero coordinates since the code $C$ has a minimum distance
$d$~\cite{Ball_geo}, thus at most $c-d$ codes fail to detect the error $e$. So
the error rate is $\epsilon=(c-d)/c=1-d/c.$\end{proof}


\begin{theorem}\label{Maintheorem} Given a $[c,2r,d]_q$ LECC, $C$, where $q=2^s$, then there is a PTP with error $\epsilon=1-d/c$ which
maps $2rs$ qubits (half held by Alice and half held by Bob) to $2(rs-s)+1$
qubits.
\end{theorem}

\begin{proof} The proof is straightforward from Theorem \ref{thm:ptc}
and the method to construct PTPs from SPTCs introduced in Section \ref{PTP}.
\end{proof}

The results in Theorem \ref{thm:ptc} (Theorem \ref{Maintheorem}) constitute an
efficient method to construct SPTCs (PTPs) from well established classical
results on coding theory. Moreover, it is straightforward to obtain the relevant
parameters of the SPTC (PTP) from the properties of the classical LECC employed in its design. In fact, it can be shown that the method proposed
in \cite{BARNUM_2002} constitutes a particular case of Theorem \ref{thm:ptc}.

\begin{corollary}\label{cor1} A maximum distance separable (MDS) code
$[q+1,2r,q+2-2r]_q$, where $q=2^s$, gives a SPTC being a collection of $q+1$
$[[rs,rs-s]]$ stabilizer QECCs with error $\epsilon=(2r-1)/(q+1)$ as described
in~\cite{BARNUM_2002}.
\end{corollary}
Note that a $[q+1,2r,q+2-2r]_q$ code is known to exist for all $2r\leq q$. One
can take as columns of $G$ the points of a normal rational
curve~\cite{Beut98,Hirs98,Casse_book}. This generates a code known as the
(extended) Reed-Solomon code~\cite{Mac77}. 

Another example of a linear code that provides
interesting parameters for the construction of a SPTC is the next
one:

\begin{corollary}\label{cor2} A LECC $[q^2+1,4,q^2-q]_q$, where $q=2^s$,
gives a SPTC being a collection of $q^2+1$ $[[2s,s]]$ stabilizer QECCs with error
$\epsilon=(q+1)/(q^2+1)$.
\end{corollary}

This linear code can be obtained from an ovoid of $PG(3,q)$.
That is, a set of $q^2+1$ points in $PG(3,q)$ with no three of them collinear.
Note that since every plane ($=$ hyperplane) of $PG(3,q)$ intersects an ovoid in
either $1$ or $q+1$ points, then the parameters of the linear code corresponding
to the ovoid are as claimed $[q^2+1,4,q^2-q]_q$~\cite{Dembowski}.

In the design of every SPTC there is always a trade-off between the number of
stabilizer QECCs within the set and the value of the error probability,
$\epsilon$, of the SPTC. In the next section, when we discuss some applications
of PTPs, we will see that the optimal choice for these two parameters depends on
the particular application. For instance, in error detection (error
purification) protocols the number of quantum codes is just translated into some
classical communication, while the parameter $\epsilon$ is the probability to
detect an error by the protocol (i.e., the fidelity of the resulting state with
respect to some maximally entangled state). In this case, therefore, it could be
important to keep $\epsilon$ as small as possible, even when the number of
quantum codes to achieve this is large \cite{Ambainis02}. In quantum message
authentication, on the contrary, the number of quantum codes in the PTP is
related with the length of the secret key (expensive resource) needed in the
protocol, while the value of $\epsilon$ measures the probability of Eve to
tamper on the line. Thus, in this scenario the decision about the value of these
two parameters depends on the security requirements~\cite{BARNUM_2002}.





\section{Applications}\label{Appl}

In this section we analyze the implications of our results in error detection,
entanglement purification and the quantum message authentication.

\subsection{Error detection}

The theory of quantum error correction was a fundamental breakthrough for
quantum information processing to become a potential feasible technology
\cite{Shor95,Steane96}. Since its conception there have been many promising
advances~\cite{QEC1,QEC2,QEC3} that might pave the way towards fault-tolerant
quantum computation but for this, a quite large number of physical qubits and
good quality gates are required~\cite{Surf_code_perspective}. However, depending
on the application needs, quantum error correction is not the only way to deal
with the possible errors that can affect a quantum system. Another interesting
approach that is particularly useful for quantum communication is that of error
detection~\cite{BARNUM_2002,Gisin04,err_det_recent}. In contrast to error
correction, where the aim is to actively correct the errors that occur during
the transmission of information, in error detection the goal is less stringent,
i.e., it is just to detect with high probability if an error has occurred. In
this last case the protocol simply discards the signal. That is, a quantum error
is transformed into an erasure. We remark that in the noisy intermediate-scale
quantum (NISQ) era, error mitigation~\cite{err_mit} is also a popular approach
to deal with errors but its main goal is to obtain the correct expectation
values of different observables even if an error has occurred in the quantum
state.  

As introduced in Section~\ref{PTP}, a PTP is basically an error detection
protocol that allows two parties to check with high probability whether or not
an error has affected the transmission of a maximally entangled state. Making
use of quantum teleportation~\cite{TeL}, therefore, one can convert PTPs into
general error detection protocols as described below. This is a different
approach to the one in~\cite{Gisin04}, where a quantum state is sent directly
to the receiver and later on it is decided whether an error has occurred and the
received state must be discarded or no error is detected and the state is
accepted. 

Here, instead, one can detect {\it a priori} the potential errors that could
have occurred during the transmission of the quantum state, without the need to
let the state be corrupted by such errors. This is particularly relevant for
quantum information since it cannot be copied. To do this, the sender uses a PTP
to send EPR pairs and check together with the receiver if an error affected the
communication process. If no error is detected then the resulting EPR pairs can
be used to teleport the desired quantum state. Otherwise, the distributed EPR
pairs are discarded. This means that the transmitter can
keep the information state and run the PTP protocol until no error is
detected and the state can be transmitted without errors. Importantly as well,
PTPs work without any conjecture about the actual error model of the channel. 

A $[c,2r,d]_q$ LECC, with $q=2^s$ provides us with the
following parameters via Theorem~\ref{Maintheorem} for the error detection
protocol. By performing the LOCC quantum operation $\mathcal{O}$ prescribed in
Definition~\ref{Def:PTP} on $2n=2rs$ qubits, the parties obtain a quantum state
of $2m=2(rs-s)$ qubits with its fidelity being at least $F=1-\epsilon=d/c$
compared to the state of $m=rs-s$ EPR pairs. This is also the resulting
fidelity of the teleported $m$ qubit message compared to the $m$ qubit state
that the parties wish to send, if one assumes an ideal teleportation process.

Let us derive the parameters $n$, $m$ and $F$ for the LECCs in
Corollaries~\ref{cor1} and~\ref{cor2}. For the linear code in
Corollary~\ref{cor1} we have $n=rs$, $m=rs-s$ and $F=(2^s+2-2r)/(2^s+1)$. And
for the linear code in Corollary~\ref{cor2} we have $n=2s$, $m=s$ and
$F=(2^{2s}-2^s)/(2^{2s}+1)$. Note that, as it has already been mentioned before,
since the purpose here is just to detect possible errors the parties do not need
to pre-share a secret key. Therefore, we can in principle set the number of
stabilizer codes ($2^s+1$ and $2^{2s}+1$ for Corollaries~\ref{cor1}
and~\ref{cor2}, respectively) as high as we wish, which means that one can
obtain fidelities arbitrarily close to 1.

We note that in~\cite{Gisin04} the authors consider an alternative method for
transferring EPR pairs between two distant parties. Precisely, to achieve a high
fidelity for the desired number of EPR pairs, one party generates a larger
number of EPR pairs than what they wish to share at the end. Then half of each
EPR pair is sent to the other party and the error detection protocol is run. If
the protocol succeeds the parties obtain the desired number of EPR pairs with
higher fidelity than what they would have obtained if they omit the error
detection protocol. A drawback of this approach is that~\cite{Gisin04} only
considers phase-shift errors. In contrast, as already mentioned,
an error detection method based on PTPs does not
require any assumption on the errors occurring in the channel.

\subsection{Entanglement purification}

Purification or distillation of
entanglement~\cite{Ben96a,Ben96,Ben96b,HorodeckiEntanglement,distLuo} is an
important operation for many quantum communication protocols. Since typical
channels are noisy, Alice and Bob usually end up with mixed entangled states,
which must then be distilled into pure ones via LOCC to make them useful for the
envisaged scheme. In general, the goal is to obtain $m$ ``high quality'' EPR
pairs from $n\geq m$ noisy ones.

The case where Alice and Bob share identical copies of the same state or,
equivalently, where the noise acts independently on each signal was addressed
in~\cite{Ben96a} for particular pure two-qubit pure entangled states and
in~\cite{Ben96,Ben96b,Horodecki97} for general mixed entangled states. The
assumption of an independent error model is justified in many communication
scenarios from technological considerations. However, there are also situations,
particularly in the cryptographic context, where the action of the channel is
controlled by Eve and in principle such an error model is not valid
anymore~\cite{Lo99,Shor00,BARNUM_2002}.

Ambainis {\it et al.}~\cite{Ambainis02} studied general entanglement
purification protocols (GEPP) within a broader error model than the one
considered in the results above, and where the previous techniques do not appear
to work. These authors no longer assume that there is a single ``distortion''
operator that acts independently on each qubit pair, neither they assume that
Alice and Bob have complete information about the distortion. The only
assumption they make is that such distortion is not very large. More
precisely, they consider that Alice and Bob share a state $\rho$ with fidelity at
least $1-\varepsilon$ with respect to a pre-defined maximally entangled state
$|\Phi^{+}\rangle^{\otimes{}n}$. Although their proposal works for arbitrary
maximally entangled states, for simplicity we restrict ourselves here to the
case of EPR pairs. 

In this context, ref.~\cite{Ambainis02} shows that is not possible to devise
absolutely successful GEPPs, that is, protocols that never fail and output a
high fidelity state. Or, to put it in other words, with these stringent requirements the parties cannot
increase the fidelity of their initial quantum state arbitrarily. However, one
can construct conditionally successful (CS)-GEPPs, in which we allow the
protocol to fail with a small probability but when it succeeds, it outputs
a quantum state with high expected fidelity (close to 1). One can also construct
so-called deterministic conditionally successful (DCS)-GEPPs, which, conditioned
on succeeding, output a high fidelity state with probability 1. The difference
between CS- and DCS-GEPPs is very subtle. We will come back to this after
stating their precise definition.

To see how our geometrical construction affects the parameters
of such CS- and DCS-GEPPs, let us recall the notation and general setting
from~\cite{Ambainis02} briefly. 

Alice and Bob possess a state in
$\mathcal{H}^A_N\otimes\mathcal{H}^B_N$ with $N$ being the dimension of the
Hilbert-spaces. Let $\ket{\Psi_N}_{AB}\in \mathcal{H}^A_N\otimes\mathcal{H}^B_N$ be
defined as
\begin{equation}
\ket{\Psi_N}_{AB}=\frac{1}{\sqrt{N}}\sum_{i=0}^{N-1}\ket{i}_A\ket{i}_B.
\end{equation}
This is a maximally entangled state in $\mathcal{H}^A_N\otimes\mathcal{H}^B_N$
as $\ket{i}_A (\ket{i}_B)$ is an orthonormal basis in $\mathcal{H}^A_N
(\mathcal{H}^B_N)$. If the dimension $N=2^n$, then the state
$\ket{\Psi_N}_{AB}$ is the state of $n$ EPR pairs. In~\cite{Ambainis02} Alice
and Bob can also be given an auxiliary input $\ket{\Psi_K}_{AB}\in
\mathcal{H}^A_K\otimes\mathcal{H}^B_K$ with dimension $K=2^k$ corresponding to
the state of $k$ EPR pairs. They utilize this state for encrypting the
classically communicated bits with a one-time pad by first distilling a secret
key from these extra $k$ EPR pairs. In the
following examples for the CS-GEPP we take $k=0$ (corresponding to
dimension $K=1$), which means that the parties do not possess extra perfect EPR
pairs in this case. The symbol $\mathcal{P}$ denotes protocols for
extracting entanglement by LOCC operations. At the end of $\mathcal{P}$, two
scenarios are possible:
\begin{enumerate}
  \item Alice and Bob abort and claim failure by outputting a special symbol
FAIL. This is denoted by $\mathcal{P}(\rho) = {\rm FAIL}$, where $\rho$ is their input
state.
  \item They output a (possibly mixed) state $\sigma\in
  \mathcal{H}^A_M\otimes\mathcal{H}^B_M$. This is denoted by $\mathcal{P}(\rho)
  = \sigma$. If $M=2^m$, then this state is
  close to $m$ EPR pairs.
\end{enumerate} 

\begin{definition*}[\cite{Ambainis02}]\label{Def:CS} A general entanglement
purification protocol ${\cal P}$ is conditionally successful (CS) with
parameters $\langle{}N,K,M,\varepsilon,\delta,p\rangle$ if for all input states
$\rho$ such that $F(\rho)=1-\varepsilon$, we have $Pr[{\cal
P}(\rho)={\rm FAIL}]\leq{}p$ and

\begin{equation}
\mathbb{E}_{{\cal P}}[F({\cal P}(\rho))|{\cal
P}(\rho)\neq{}{\rm FAIL}]\geq{}1-\delta,
\end{equation}

\noindent where $\mathbb{E}_{{\cal P}}$ denotes the expectation taken over the
classical communication in the protocol ${\cal P}$.
\end{definition*}\noindent By the fidelity of a state $\rho\in\mathcal{H}^A_N\otimes\mathcal{H}^B_N$ we
mean:
\begin{equation}
F(\rho)=\bra{\Psi_N}\rho\ket{\Psi_N}_{AB}.
\end{equation}If 
$\sigma\in\mathcal{H}^A_M\otimes\mathcal{H}^B_M$ then
$F(\sigma)=\bra{\Psi_M}\sigma\ket{\Psi_M}.$

Note that the parameters $p$ and $\delta$ depend on $\varepsilon$, $N$ and $M$,
however, to be consistent with the notation in~\cite{Ambainis02} we suppress
this dependence. Also note that in the definition above, the requirement is only
that the fidelity averaged over all the possible classical communication
scenarios between the parties should be high when the protocol succeeds. A
successful protocol can be obtained by conducting the classical communication
between the parties in many different ways and Definition~\ref{Def:CS} takes the
average of the fidelities of the resulting states in all these different
classical communication scenarios. However, as described in~\cite{Ambainis02} it
is possible (with small probability) that the actual fidelity is much lower even
if the protocol succeeds. This is because a possible eavesdropper who sees all
the classical communication can use this information to attack Alice and Bob
since she can know the fidelity of the accepted state. 

To address this adversarial setting one can use a stronger definition (DCS) that
requires that, in the case of success, regardless of the classical messages
interchanged between Alice and Bob, they obtain a high fidelity state. For this,
the key idea is to simply encrypt the classical communication. In doing so, Eve
cannot infer information about the fidelity of the resulting state. 

\begin{definition*}[\cite{Ambainis02}]\label{Def:DCS} A general
entanglement purification protocol ${\cal P}$ is deterministically conditionally
successful (DCS) with parameters $\langle{}N,K,M,\varepsilon,\delta,p\rangle$ if
for all input states $\rho$ such that $F(\rho)=1-\varepsilon$, we have $Pr[{\cal
P}(\rho)={\rm FAIL}]\leq{}p$ and

\begin{equation}
Pr\big[F({\cal P}(\rho))\geq{}1-\delta{}\,|\,{\cal P}(\rho)\neq{}{\rm FAIL}\big]=1.
\end{equation}
\end{definition*}

We remark that a DCS-GEPP is also a CS-GEPP and one can construct a DCS-GEPP
from a CS-GEPP with the help of additional EPR pairs, used to encrypt the classical
communication with a one-time pad.

\begin{fact*}[\cite{Ambainis02}]\label{claim1} A CS-GEPP protocol with
parameters $\langle{}N,K,M,\varepsilon,\delta,p\rangle$ which uses $c$ bits of
communication can be converted to a DCS-GEPP protocol with parameters
$\langle{}N,2^{c}K,M,\varepsilon,\delta,p\rangle$.
\end{fact*}

Ref.~\cite{Ambainis02} shows that a PTP with error
$\epsilon$ naturally gives a GEPP since one can just run the PTP, outputting FAIL
whenever the PTP rejects the input. Considering that we have chosen $N = 2^n$,
$M = 2^m$ and $k=0$, we obtain a CS-GEPP with parameters 
\begin{equation}\label{CSRes}
\left\langle2^n, 1, 2^{n-s},\varepsilon,\frac{\epsilon}{1-\varepsilon}, \varepsilon\right\rangle,
\end{equation}
for any $s\in\lbrace 1,\dots,n \rbrace$. This means that Alice and Bob start
with a state with fidelity $1-\varepsilon$ compared to $n$ EPR pairs and at the
end of the process $\mathcal{P}$ they obtain a state that has on average (in the
sense of Definition~\ref{Def:CS}) a fidelity of $1-\epsilon/(1-\varepsilon)$
with respect to the state of $n-s$ EPR pairs. 

By Theorem~\ref{Maintheorem}, a LECC
$[c,2r,d]_{q}$, with $q=2^s$ gives an SPTC and thus a PTP with error
$\epsilon=1-d/c$, which maps $2rs$ qubits (half held by Alice and half held by
Bob) to $2(rs-s)+1$ qubits. This PTP requires a classical communication of
$b=\ceil{\log_2(c)}+s$ bits, since the parties need to choose a stabilizer code
randomly out of the $c$ stabilizers possible and communicate which one they are
using. Moreover, they also need to compare the syndrome of the chosen
stabilizer, which is an $s$-bit string.

Let us apply the LECC $[q+1,2r,q+2-2r]_q$, with $q=2^s$
in Corollary~\ref{cor1} to Eq.~\eqref{CSRes}. This is the
scenario considered in~\cite{Ambainis02}. This means that $n=rs$
and $\epsilon=(2r-1)/(2^s+1)$ thus we have a CS-GEPP with the following
parameters
\begin{equation}\label{CSRes1}
\left\langle2^{rs}, 1, 2^{rs-s},\varepsilon,\frac{2r-1}{(1-\varepsilon)(2^s+1)}, \varepsilon\right\rangle.
\end{equation}
In this case Alice and Bob need to communicate
$c=\ceil{\log_2(2^s+1)}+s=2s+1$ classical bits. Therefore, by Fact~\ref{claim1},
we also obtain a DCS-GEPP with the following parameters
\begin{equation}\label{DCSRes1}
\left\langle2^{rs}, 2^{2s+1}, 2^{rs-s},\varepsilon,\frac{2r-1}{(1-\varepsilon)(2^s+1)}, \varepsilon\right\rangle.
\end{equation}

Now, let us consider the LECC $[q^2+1,4,q^2-q]_q$, with
$q=2^s$ from Corollary~\ref{cor2}. Similarly, as before, this
means that $n=2s$ and $\epsilon=(2^s+1)/(2^{2s}+1)$ and thus we have a
CS-GEPP with the parameters
\begin{equation}\label{CSRes2}
\left\langle2^{2s}, 1, 2^{s},\varepsilon,\frac{2^s+1}{(1-\varepsilon)(2^{2s}+1)}, \varepsilon\right\rangle.
\end{equation}
In this case Alice and Bob need to communicate
$b=\ceil{\log_2(2^{2s}+1)}+s=3s+1$ classical bits. Therefore, by
Fact~\ref{claim1}, we also obtain a DCS-GEPP with the following parameters
\begin{equation}\label{DCSRes2}
\left\langle2^{2s}, 2^{3s+1}, 2^{s},\varepsilon,\frac{2^s+1}{(1-\varepsilon)(2^{2s}+1)}, \varepsilon\right\rangle.
\end{equation}

We now compare the two constructions for the CS-GEPP and DCS-GEPP given by
Eqs.~\eqref{CSRes1}-\eqref{DCSRes2}. For this to be a fair comparison we require
$n$ and $m$ to be the same in both cases, which only happens if we set $r=2$,
which implies that $n=2s$ and $m=s$. We also fix the value of the initial
fidelity $1-\varepsilon$ for all the four cases. The results are shown in
Fig.~\ref{fig:GEPPComparison}. 

\begin{figure}[H]
  \centering
  \subfigure[]{\includegraphics[scale=0.5]{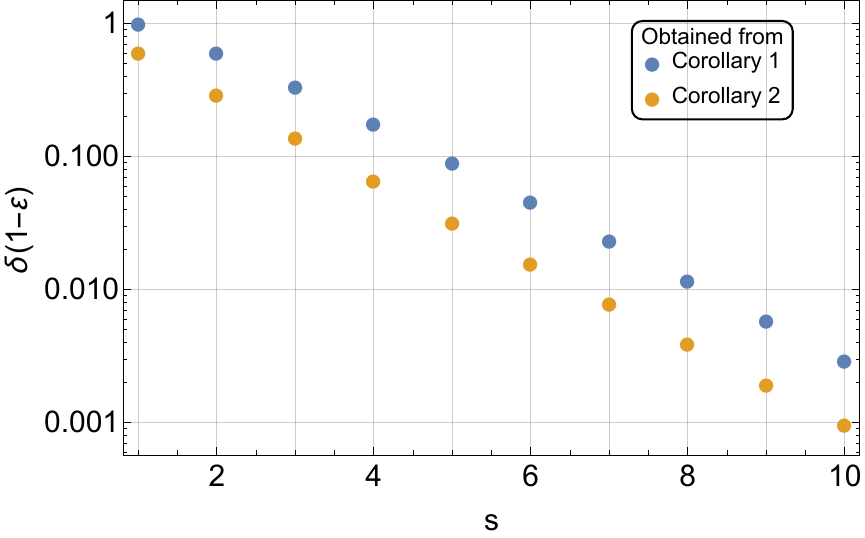}}
  \label{fig:GEPPComparisona}
  \subfigure[]{\includegraphics[scale=0.5]{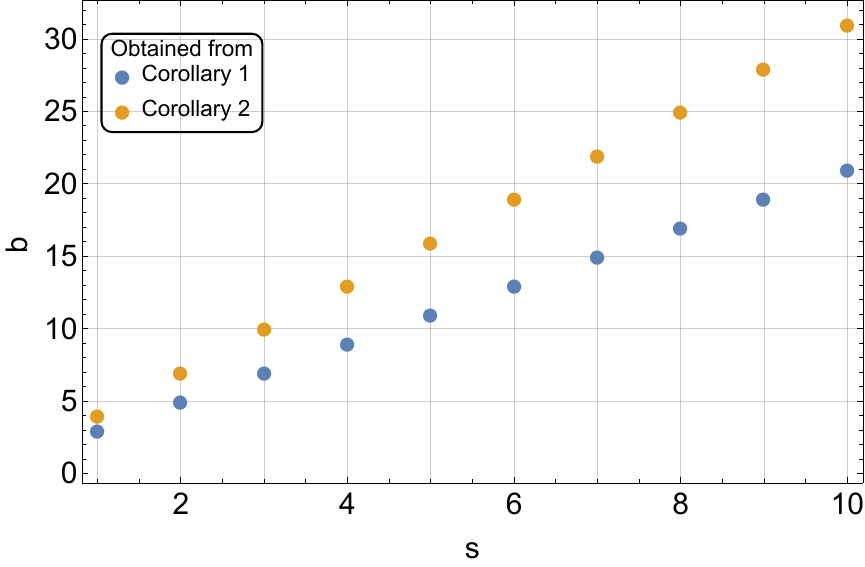}}
  \label{fig:GEPPComparisonb} 
    \caption{(a) The quantity $\delta(1-\varepsilon)$ as a function of $s$,
    where $\delta$ is the deviation from 1 in the fidelity of the output
    state in the case of acceptance for the CS-GEPP (lower
    $\delta(1-\varepsilon)$ values mean higher fidelity). The quantity
    $(1-\varepsilon)$ is the fidelity of the initial state, which is fixed for both
    Corollaries. (b) The number of classically communicated bits $b$ as a function
    of $s$. Note that $b$ is also the number of the auxiliary perfect EPR pairs
    required when transforming the CS-GEPP to DCS-GEPP via Fact~\ref{claim1}. }
    \label{fig:GEPPComparison}
  \end{figure}

Precisely, in
Fig.~\ref{fig:GEPPComparison}a we see that Corollary~\ref{cor2} performs better
in terms of average fidelity of the final state in the case of the CS-GEPP. When
converting the corresponding CS-GEPPs to DCS-GEPPs with Fact~\ref{claim1} the
final fidelity stays the same, but the DCS-GEPP obtained from
Corollary~\ref{cor2} requires more auxiliary perfect EPR pairs (see
Fig.~\ref{fig:GEPPComparison}b) than the DCS-GEPP obtained from
Corollary~\ref{cor1}. Therefore, when designing a CS-GEPP Corollary~\ref{cor2}
provides better results. However, when constructing a DCS-GEPP the best choice
depends on the application since the Corollary providing better final fidelity
necessitates a higher number of perfect EPR pairs.

\subsection{Authentication of quantum messages}

A quantum authentication scheme (QAS) is a cryptographic protocol in which Alice
wishes to send a quantum state to Bob in a way that he can be sure that the
state has arrived unaltered from Alice. For such a task, Alice and Bob
need a set of keyed encoding and decoding maps. Alice chooses and applies such a
map to the state she wishes to send based on the secret key that she shares
with Bob. Bob then applies the corresponding decoding map to the received state
and rejects or accepts the resulting quantum state as authentic based on the
state of a flag qubit. Let us now recall the formal definition
from~\cite{BARNUM_2002,reuse}:

\begin{definition*}[\cite{BARNUM_2002,reuse}]\label{secParam} A
non-interactive QAS with error $\epsilon$ is a set of classical keys
$\mathcal{K}$ such that for $\forall k \in \mathcal{K}$, $\exists A_k, B_k$
computable trace preserving completely positive (TPCP) maps. The TPCP map $A_k$
takes an $m$-qubit message state $\rho$, adds $s$ auxiliary qubits and outputs a
system $\rho_k$ of $m+s$ qubits. $B_k$ takes the (possibly altered) state
$\rho'_k$ as input and outputs two quantum systems, an $m$-qubit state $\rho'$
and a single-qubit state with basis states $\ket{{\rm ACC}}$ and $\ket{{\rm REJ}}$
indicating acceptance or rejection. Moreover, for all message states
$\ket{\psi}$ we have:
\begin{itemize}
\item Completeness: $\forall k \in \mathcal{K}$:
$B_k(A_k(\ketbra{\psi}{\psi}))=\ketbra{\psi}{\psi}\otimes \ketbra{{\rm ACC}}{{\rm ACC}}$	
\item Soundness: $\Tr(P\rho')\geq 1-\epsilon$, where
$P=\ketbra{\psi}{\psi}\otimes \ketbra{{\rm ACC}}{{\rm ACC}}+\mathds{1}_m\otimes
\ketbra{{\rm REJ}}{{\rm REJ}}$, with $\mathds{1}_m$
being the identity operator on the message space.
\end{itemize}
\end{definition*}

We note that similarly to the definition of PTPs, for a QAS to be considered
secure we only require that the protocol accepts an altered state with a tiny
probability $\epsilon$. It is also important to note that in this definition
only pure states are considered as messages. This is due to the fact that
in~\cite{BARNUM_2002} the authors only demonstrate the security of the non-interactive
QAS in Protocol~\ref{prot1} (see below) for pure message states. They prove that
by transforming Protocol~\ref{prot1} into an interactive QAS that achieves
quantum teleportation~\cite{BARNUM_2002}. Then, by linearity, one expects that
the security extends trivially for mixed state messages as well. But this
statement has only been proven rigorously in~\cite{HLM}, where the authors show
that the non-interactive QAS in Protocol~\ref{prot1} stays secure under the more
general definition of security where mixed state messages are also allowed. This
can be accounted for in Definition~\ref{secParam} by regarding $\ket{\psi}$ as a
purification of the mixed state message. This means that formally
Definition~\ref{secParam} is still adequate as a security definition. 

There have been many proposals for non-interactive QASs, like e.g. the Clifford
code~\cite{Aharonov,TrapAndClifford} and the trap
code~\cite{trap,TrapAndClifford}. In the case of the Clifford code based QAS
Alice appends a certain number of qubits prepared in the state $\ket{0}$ to her
message and then applies a certain, random Clifford operation~\cite{Aharonov}
corresponding to a secret key shared with Bob. On the receiving side, Bob
performs the inverse Clifford operation, measures the state of the appended
qubits and accepts the state if he finds that they are in the $\ket{0}$ state.
On the other hand, the trap code based QAS requires that Alice first encodes one
qubit with a quantum error correcting code with distance $d$ into $n$ qubits.
Then she appends $n$ qubits in the state $\ket{0}$ and $n$ qubits in the
$\ket{+}$ state. After this, she applies a random permutation on the $3n$
qubits, which is indexed by a shared secret key. Finally, she encrypts the
resulting state using the quantum one-time pad~\cite{qotp}, also using a part of
the shared secret key. Bob performs the inverse of the applied quantum one-time
pad and the inverse of the permutation and accepts the message if the last $2n$
qubits are in their initial state. In such case he decodes the first $n$ qubits
according to the agreed quantum error correcting code. 

All the above-mentioned schemes have security parameters (see
Definition~\ref{secParam}) that are exponentially small in some tunable
parameter of the corresponding scheme. As we have noted before, for all QASs it
is necessary that the parties pre-share a secret key. This is an expensive
resource, therefore when comparing QASs it is not enough to just compare the
security parameter but we have to evaluate the number of necessary secret key
bits and also the number of auxiliary qubits. In this regard, the trap and
Clifford codes have certain drawbacks. For example with the trap code we have to
start with one message qubit and use a large number of auxiliary qubits. In the
Clifford code, on the other hand, we have to store the index of specific Clifford operations as a
shared secret key which requires large number of secret key bits. Therefore,
below we focus on the QAS based on SPTCs, considered in~\cite{BARNUM_2002}
since, besides having a security parameter that can also be made exponentially
small, it gives a better flexibility in tuning the security parameter and the
number of shared secret bits required. This gives us more freedom in finding the optimal construction for our
purposes.

The previously mentioned non-interactive QAS~\cite{BARNUM_2002}
obtained from SPTCs runs as follows.   

\begin{protocol*}[\cite{BARNUM_2002}]\label{prot1}{\textcolor{white}{.}\\
1. Alice and Bob share a secret key $x$ of length $2m$ to be used for
q-encryption (using the quantum one-time pad~\cite{qotp}). For authentication, they
additionally agree on a SPTC $\{Q_k\}$ and two secret binary strings $k$ and
$y$. \\\\
2. Alice q-encrypts the message state $\rho$ of $m$ qubits as
$\rho_e=\sigma_x^{x_1}\sigma_z^{x_2}\rho
\sigma_z^{x_2}\sigma_x^{x_1}$, where $x_1$ ($x_2$) is the
first (second) half of the $2m$-bit string $x$. Next, Alice encodes $\rho_e$
according to $Q_k$ in a way that the syndrome is $y$ (there
is freedom in choosing the logical basis states in the stabilizer space
therefore one can choose the syndrome) to produce $\sigma$. Since $Q_k$ encodes
$m$ qubits into $n=m+s$ qubits $\sigma$ is a quantum state of $m+s$ qubits.
Finally, Alice sends the result to Bob. \\\\
3. Bob receives the $n=m+s$ qubits. Denote the received state by $\sigma'$. He
measures the syndrome $y'$ of the code $Q_k$ on his qubits. Bob compares $y$ to
$y'$, and aborts if any error is detected. Otherwise, he decodes his $n$-qubit
word according to $Q_k$, obtaining $\rho_e'$. Bob q-decrypts $\rho_e'$ using $x$
and obtains $\rho'$.}\end{protocol*}

We remark that in contrast to the classical case, encryption is required for the
authentication of quantum messages (this is labeled as q-encrypt in step 2 of
Protocol~\ref{prot1}), as it has been shown in~\cite{BARNUM_2002}. Moreover,
encoding with a stabilizer is required so that it is possible to detect errors
on the quantum state. 

Note that the security parameter $\epsilon$ of Protocol~\ref{prot1} coincides
with the error of the SPTC used for its construction. Let us denote the
length of a binary string $g$ by $|g|$. With this notation, the required length
of the shared secret key for Protocol~\ref{prot1} is
$l=2m+|k|+|y|=2m+\ceil{\log_2(K)}+s$, where $K$ is the number of stabilizer codes in the SPTC and $|y|=s$,
since this is the number of bits in the syndrome of each stabilizer code in the
SPTC. 

By the general results presented in Section~\ref{Geo_without_barnum}, a
LECC $[c,2r,d]_{q}$, with $q=2^s$ gives a SPTC and thus
a non-interactive QAS with an error $\epsilon=1-d/c.$ Moreover, it
requires a shared secret key of length $l=2(rs-s)+\ceil{\log_2(c)}+s$ and
encodes an $rs-s$ qubit message into an $rs$ qubit state. This means that
one can only obtain useful QASs if $r>1,$ otherwise the message consists of
zero qubits.

In the case of Corollary~\ref{cor1}, it provides a
QAS with an error:
\begin{equation}
\epsilon=\frac{2r-1}{2^s+1},
\end{equation}
and the necessary key length is
\begin{equation}
l=2(rs-s)+\ceil{\log_2(2^s+1)}+s=2rs+1.
\end{equation}
By choosing $r$ and $s$ we can tailor the QAS protocol as desired. This example
has been considered in~\cite{BARNUM_2002}. 


In the case of Corollary~\ref{cor2}, it provides a
QAS with an error:
\begin{equation}
\epsilon=\frac{2^{s}+1}{2^{2s}+1},
\end{equation}
and the necessary key length is
\begin{align}
l&=2(2s-s)+\ceil{\log_2(2^{2s}+1)}+s=\nonumber\\
&=3s+\ceil{\log_2(2^{2s}+1)}=5s+1.
\end{align}
Note that
increasing $s$ makes the error smaller and increases the
required length of the shared secret key linearly. 

To compare how the two different families of LECCs from
Corollaries~\ref{cor1} and~\ref{cor2} perform in the QAS given by
Protocol~\ref{prot1} we consider the practical tasks of authenticating a quantum
message consisting of $\sim 10^5$ and $\sim 10^2$ qubits. As mentioned before,
due to Theorem~\ref{thm:ptc}, a LECC $[c,2r,d]_q$ with $q=2^s$ and $c$,
$r$, $d$ and $s$ being positive integers gives an SPTC consisting of stabilizers
encoding an integer number of qubits ($rs-s$) into $rs$ qubits. This means that
we have to perform Protocol~\ref{prot1} in blocks of size $rs-s$. This is the
reason why we say that there are approximately $10^5$ ($10^2$) qubits in the
message as to be able to use the same SPTC for all the blocks, we cannot have
remainder qubits. Alternatively, one could also complete the last block with,
say, $\ket{0}$ states. Therefore, we will select the total number of qubits in
the message to be divisible by $rs-s$. This way we can perform
Protocol~\ref{prot1} via the SPTC that the code $[c,2r,d]_q$ gives in all the
$\sim 10^5/(rs-s)$ and $\sim 10^2/(rs-s)$ blocks. So if we consider the first
case and take a specific LECC $[c,2r,d]_q$ with $q=2^s$ then we have
$B=\floor{10^5/(rs-s)}$ blocks with $\epsilon=1-d/c$ being the error of the QAS
for each block. This means that $sB$ is the number of necessary auxiliary qubits
and $B\big[2(rs-s)+\ceil{\log_2(c)}+s\big]$ is the number of secret key bits necessary
for the quantum authentication task for the whole quantum message. Due to the
fact that we have to deal with possibly many blocks we introduce a new
quantity, namely an upper bound $\epsilon_{\rm{total}}$ on the probability of
having an error in at least one of the blocks. More precisely this is an upper
bound on the probability that the whole message is accepted as authentic
but actually there is at least one error. This quantity can be calculated
as\begin{equation} \epsilon_{\rm{total}}=1-(1-\epsilon)^{B}, \end{equation}
since $\epsilon$ is the error in each of the $B$ blocks. The goal is to keep
$\epsilon_{\rm{total}}$ small. 
For this, $\epsilon$ has to be small, which requires the distance $d$ of the
code to be as close to $c$ as possible. This means that for a given $c$ and $r$
the smallest total error is achieved by codes that saturate the Singleton bound
($d=c-2r+1$). We also remark that to have small $\epsilon$ we have to choose $s$
to be relatively large (i.e., we have to go beyond binary codes). Intuitively,
this is because $s$ is the difference between the number of qubits that we
encode and the number of qubits that we encode into for each stabilizer in the
SPTC. Therefore, a larger $s$ gives a larger difference in the dimensions of the
Hilbert spaces, which means we have more possibilities for choosing different
stabilizers for the SPTC. In this way, it will be sensitive for more Pauli
errors. We take the two families of codes from Corollaries~\ref{cor1}
and~\ref{cor2} with different parameters and summarize the relevant quantities
from above in Tables~\ref{tab:100000} and~\ref{tab:100}. We expect
Corollary~\ref{cor1} to perform better for more parameters since it saturates
the Singleton bound and, as opposed to Corollary~\ref{cor2}, the dimension of
the code can be adjusted independently of the parameter $s$.

The results for the $\sim 10^5$ qubit case can be seen in
Table~\ref{tab:100000}. Since in this scenario the number of qubits in the
message is relatively large, in order to keep $\epsilon_{\rm{total}}$ small we
have to choose the error in each block to be tiny and the number of blocks low
(i.e., the block length large). From Table~\ref{tab:100000}, we see that
Corollary~\ref{cor2} can perform better in terms of error when considering error
correcting codes over the same order field $GF(2^s)$. However, this comes at the
price of other parameters being worse compared to Corollary~\ref{cor1}.
Corollary~\ref{cor1}, on the other hand, gives more freedom in setting the
relevant parameters of the QAS since with increasing $r$ we can increase the
block size without increasing the order of the underlying field. This results in
a lower number of auxiliary qubits and a lower number of secret key bits but a
slightly higher error rate. Moreover, the number of blocks also decreases,
therefore the value of $\epsilon_{\rm{total}}$ does not necessarily get worse in
this case. However, with increasing $r$ it becomes more difficult to set the
number of message qubits to $\sim 10^{5}$. Upon designing the QAS one has to
decide which properties are more important and choose the underlying LECC
accordingly. 

The results for the case of a quantum message consisting of a much lower
number of qubits (i.e., $\sim 10^2$ qubits) can be found in
Table~\ref{tab:100}. The tendencies that we have observed regarding
Table~\ref{tab:100000} stay valid in this case as well. Here we do not have a
large number of qubits in the message so $s$ ($2^s$ is the order of the field
used) does not need to be as large as for the case in Table~\ref{tab:100000} to
have a relatively small $\epsilon_{\rm{total}}$. The only constraint is
that the number of message qubits has to be close to a multiple of the size of
the blocks $(rs-s)$, therefore, depending on $s,$ one cannot allow for larger
values of $r,$ which was possible for the case of Table~\ref{tab:100000}.

As a conclusion it is clear that increasing $s$ decreases the value of
$\epsilon_{\rm{total}}$ the most effectively for both Corollaries. On the other
hand, the value of $r$ for Corollary~\ref{cor1} can help decreasing the amount
of other required resources, like the length of the secret key or the number of auxiliary
qubits.

Finally, since the shared secret key between Alice and Bob is an expensive
resource it is worth noting that QASs with secret key recycling have been
devoted a lot of attention~\cite{reuse,HLM,Portmann, garg}, recycling means that
some portion of the key can be reused in a subsequent round of the QAS.
In~\cite{Portmann} it is proven that in a QAS constructed from a strong SPTC
(see definition below) every bit of the secret key can be reused in the case of
acceptance and that only the bits used for the quantum one-time pad have to be
disposed when the message is rejected. Portmann~\cite{Portmann} used the most
general security definition, that is, the adversary can possess the purification
of the quantum message to be authenticated and the consideration is not
restricted to  substitution attacks as in~\cite{HLM}. However, as mentioned
before, for this we need the SPTC to be strong, which means that it is required
that all non-identity Pauli-errors have to be detected with high probability,
not just the ones that do not act trivially on the message. Note that
in~\cite{Portmann} it is also shown that the SPTC in Corollary~\ref{cor1} is
actually a strong SPTC, therefore in the accept case all the secret key bits can
be reused.  


Importantly, Theorem~\ref{thm:ptc} gives a strong
SPTC since we count a Pauli error in the centralizer of the code, i.e., in
$Q_k^{\perp}$ and not in $Q_k^{\perp} \setminus Q_k$. So this means that when an
SPTC obtained from Theorem~\ref{thm:ptc} is used for quantum message
authentication then the secret key can be reused with the following conditions.
Namely, in the case of accepting the quantum message one can reuse the whole
secret key used in the QAS and only the encryption part of the key has to be
discarded in the case of rejection~\cite{Portmann}.
\onecolumngrid
\begin{center}
  \begin{table}[h]
      \begin{tabular} {|c|c|c|c|c|c|c|}
        \hline
        Code & \thead{Number of qubits\\ in the message} & \thead{Number of\\
        auxiliary qubits} & \thead{Number of necessary\\ secret key bits} &
        \thead{Block length}& $\epsilon$ & $\epsilon_{\rm{total}}$ \\ \hline
        Corollary 2 with $s=15$ & 99 990& 99 990 & 506 616 & 15 & 3.0519 $\cdot$
        $10^{-5}$ & 0.1841 \\ 
        Corollary 1 with $r=2$ and $s=15$ & 99 990 & 99 990 & 406 626 & 15 &
        9.155 $\cdot$ $10^{-5}$  & 0.4568\\ 
        Corollary 1 with $r=4$ and $s=15$ & 99 990 & 33 330 & 268 862 & 45 &
        2.1362 $\cdot$ $10^{-4}$ &0.3779\\
        Corollary 1 with $r=11$ and $s=15$ & 99 990 & 9 990 & 220 446 & 150 &
        6.4085 $\cdot$ $10^{-4}$ &0.3475\\ \hline
        Corollary 2 with $s=25$ & 100 000& 100 000 & 504 000 & 25 & 2.9802 $\cdot$
        $10^{-8}$ & $1.192\cdot{10^{-4}}$ \\ 
        Corollary 1 with $r=2$ and $s=25$ & 100 000 & 100 000 & 404 000 & 25 &
        8.9407 $\cdot$ $10^{-8}$  & 3.5756 $\cdot{10^{-4}}$\\ 
        Corollary 1 with $r=4$ and $s=25$ & 99 975 & 33 325 & 267 933 & 75 &
        2.0862 $\cdot$ $10^{-7}$ &2.7805 $\cdot 10^{-4}$\\
        Corollary 1 with $r=11$ and $s=25$ & 100 000 & 10 000 & 220 400 & 250 &
        6.2585 $\cdot$ $10^{-7}$ &2.5031 $\cdot 10^{-4}$\\ \hline
        Corollary 2 with $s=35$ & 99 995& 99 995 & 502 832 & 35 & 2.9104 $\cdot 
        10^{-11}$ & 8.315 $\cdot 10^{-8}$ \\ 
        Corollary 1 with $r=2$ and $s=35$ & 99 995 & 99 995 & 402 837 & 35 &
        8.7312 $\cdot$ $10^{-11}$  & 2.4945 $\cdot$ $10^{-7}$\\ 
        Corollary 1 with $r=4$ and $s=35$ & 99 960 & 33 320 & 267 512 & 105 &
        2.0373 $\cdot$ $10^{-10}$ & 1.9395 $\cdot$ $10^{-7}$\\
        Corollary 1 with $r=11$ and $s=35$ & 99 750 & 9 975 & 219 735 & 350 &
        6.1118 $\cdot$ $10^{-10}$ & 1.7419 $\cdot$ $10^{-7}$\\ \hline
    \end{tabular} 
    \caption{The properties of the quantum message authentication scheme in
    Protocol~\ref{prot1} for $\sim 10^5$ message qubits with SPTCs
    obtained via Theorem~\ref{thm:ptc} from different LECCs.}
    \label{tab:100000}   
  \end{table}
\end{center}
\begin{center}
\begin{table}[H]
    \begin{tabular} {|c|c|c|c|c|c|c|}
        \hline
        Code & \thead{Number of qubits\\ in the message} & \thead{Number of\\
        auxiliary qubits} & \thead{Number of necessary\\ secret key bits} &
        \thead{Block length}& $\epsilon$ & $\epsilon_{\rm{total}}$ \\ \hline
        Corollary 2 with $s=10$ & 100 & 100 & 510  & 10 & 0.001 & 0.0097 \\ 
        Corollary 1 with $r=2$, $s=10$ & 100  & 100 & 410 & 10 & 0.0029 &
        0.0289\\ 
        Corollary 1 with $r=4$, $s=10$ & 90 & 30 & 243 & 30 & 0.0068 &0.0203 \\
       \hline
        Corollary 2 with $s=15$ & 90 & 90 & 456 & 15 & $3.0519\cdot10^{-5}$ &
        $1.831\cdot 10^{-4}$ \\ 
        Corollary 1 with $r=2$, $s=15$ & 90  & 90 & 366 & 15 & $9.155\cdot
        10^{-5}$  &  $5.4917\cdot10^ {-4}$\\ 
        Corollary 1 with $r=4$, $s=15$ & 90  & 30 & 242 & 45 & $ 2.1362\cdot
        10^{-4}$ & $4.2718\cdot 10^{-4}$\\ \hline
        Corollary 2 with $s=30$ & 90 & 90 & 453 & 30 & $9.3132\cdot10^{-10}$ &
        $2.794\cdot 10^{-9}$ \\
        Corollary 1 with $r=2$, $s=30$ & 90  & 90 & 363 & 30 & $2.794\cdot
        10^{-9}$  &  $8.382\cdot10^ {-9}$\\ 
        Corollary 1 with $r=4$, $s=30$ & 90  & 30 & 241 & 90 & $ 6.5193\cdot
        10^{-9}$ & $6.5193\cdot 10^{-9}$\\ \hline
    \end{tabular} 
    \caption{The properties of the quantum message authentication scheme in
    Protocol~\ref{prot1} for $\sim 10^2$ message qubits with SPTCs obtained via
    Theorem~\ref{thm:ptc} from different LECCs.}
    \label{tab:100}   
\end{table}
\end{center}
\twocolumngrid

\section{Conclusion}\label{Concl} We have introduced a method to obtain
stabilizer purity testing codes (SPTCs) directly from classical linear error
correcting codes (LECCs). This provides a systematic way of obtaining SPTCs and
thus, also purity testing protocols (PTPs), which can decide with high
probability if a quantum state is close to a certain number of EPR pairs. Then,
for illustration purposes, we have evaluated the performance of the PTPs
constructed from two different families of LECCs for different quantum
communication applications, including error detection, entanglement purification
and quantum message authentication. For entanglement purification, we have
considered two different types of entanglement distillation protocols
introduced in~\cite{Ambainis02}. We found that different families of LECCs can be better
in optimizing a certain parameter of the protocols but for other parameters we
might need to resort to other families of LECCs. In the case of quantum message
authentication our method can be considered as a generalization of that
introduced in~\cite{BARNUM_2002} in the sense that we also use ideas from
projective geometry but it makes it possible to obtain more families of SPTCs
with parameters that can be tuned more flexibly. In this regard, we also found
that depending on the parameter of interest (which can be the number of secret
key bits or the error parameter of the scheme), it might be advantageous to
consider different families of LECCs. In this regard, our method gives more room to
engineer the parameters of the quantum authentication schemes compared
to~\cite{BARNUM_2002}, which we showed to originate from a particular LECC.
Importantly, our construction gives strong SPTCs, which means that they are
guaranteed to have good secret key recyclability properties. Most importantly,
it provides the possibility to tune the parameters of the
above-mentioned protocols further by using different families of LECCs beyond
the two that we have tested.

\section{ACKNOWLEDGEMENTS}
D.G. and M.C. thank the University Friedrich-Alexander Erlangen-Nürnberg for its
hospitality when this project was first conceived. D.G. was supported by an
Alexander von Humboldt Fellowship during his stay in Erlangen. R.T. acknowledges
the support of the EU (QuantERA MENTA), the Spanish MCIU (Grant No.
PCI2022-132947), the Basque Government (Grant No. IT1470-22). R.T. acknowledges
the support of the Grant No. PID2021-126273NB- I00 funded by
MCIN/AEI/10.13039/501100011033 and by ``ERDF A way of making Europe''. R.T. also
thanks the ``Frontline'' Research Excellence Programme of the NKFIH (Grant No.
KKP133827) and also the Project No. TKP2021-NVA-04, which has been implemented
with the support provided by the Ministry of Innovation and Technology of
Hungary from the National Research, Development and Innovation Fund, financed
under the TKP2021-NVA funding scheme. R.T. thanks the National Research,
Development and Innovation Office of Hungary (NKFIH) within the Quantum
Information National Laboratory of Hungary. S.B. acknowledges the grant
PID2023-147202NB-I00 of the Spanish Ministry of Science, innovation and
universities. M.C. acknowledges the
support by the Galician Regional Government (consolidation of Research Units:
AtlantTIC), the Spanish Ministry of Economy and Competitiveness (MINECO), the
Fondo Europeo de Desarrollo Regional (FEDER) through the grant No.
PID2020-118178RB-C21, MICIN with funding from the European Union
NextGenerationEU (PRTR-C17.I1) and the Galician Regional Government with own
funding through the ``Planes Complementarios de I+D+I con las Comunidades
Autónomas'' in Quantum Communication, the European Union’s Horizon Europe
Framework Programme under the Marie Sklodowska-Curie Grant No. 101072637
(Project QSI) and the project ``Quantum Security Networks Partnership'' (QSNP,
grant agreement No 101114043).
\nocite{*}
\appendix
\section{The companion matrix formalism for
$GF(2^s)$}\label{app:companion_matrix} We represent the elements of $GF(2^s)$ as
matrices with the help of a monic irreducible polynomial of degree $s$ over
$GF(2)$ which is a minimal polynomial of a generator of $GF(2^s)$. In other
words, we use a primitive polynomial of degree $s$ over $GF(2).$ 

In this way, the
additive and the multiplicative structure of $GF(2^s)$ become matrix
addition and multiplication, respectively. Let such a primitive polynomial over
$GF(2)$ be $c_0+c_1 x+\cdots +c_{s-1}x^{s-1}+x^s$ with $c_i\in GF(2)$. Then, the
companion matrix $C$ of this polynomial can be written as
\begin{equation}C = \begin{pmatrix} 
  0 & 0 & 0 & \dots \,0 & -c_0 \\
  1 & 0 & 0 & \dots \,0 & -c_1 \\
  0 & 1 & 0 & \dots \,0 & -c_2 \\
  \vdots & \vdots & \vdots & & \vdots \\
  0 & 0 & 0 & \dots \,1 & -c_{s-1}
  \end{pmatrix}.
\end{equation}
The elements of the field $GF(2^s)$ can be obtained as follows. The 0 element of
$GF(2^s)$ corresponds to the $s\times s$ zero matrix. The remaining $2^s-1$
elements are generated via the powers of $C$, so they can be listed as
$C,C^2,...,C^{2^s-1},$ where we note that $C^{2^s-1}=\mathds{1}_{s\times s}$ is
the $s\times s$ identity matrix that corresponds to the multiplicative
identity element (i.e., 1) of $GF(2^s).$ 

In particular, and for illustration purposes, let us represent
$GF(4)$ as matrices with the companion matrix method. The polynomial $x^2+x+1$
over $GF(2)$ is primitive. Thus, the $C$ matrix can be written as 
\begin{equation}
  C=\begin{pmatrix} 
    0 & 1  \\
    1 & 1
  \end{pmatrix},
\end{equation}
where we use the fact that the characteristic of $GF(2)$ is 2, thus $-1=1$. Using the
method described above, we can list the elements of $GF(4)$, represented as
matrices, as follows 
\begin{equation}\label{eq:companion_GF4}
  GF(4)=\left\{\begin{pmatrix} 
    0 & 0  \\
    0 & 0
  \end{pmatrix},\begin{pmatrix} 
    1 & 0  \\
    0 & 1
  \end{pmatrix},\begin{pmatrix} 
    0 & 1  \\
    1 & 1
  \end{pmatrix},\begin{pmatrix} 
    1 & 1  \\
    1 & 0
  \end{pmatrix}\right\}.
\end{equation} 
The elements of the $GF(4)$ field are usually represented as  
\begin{equation}\label{eq:normal_GF4}
  GF(4)=\{0,1,\mu,\mu+1\},
\end{equation} 
such that $\mu^2=\mu+1.$ Therefore, we have the following correspondence
between the two descriptions 
\begin{equation}\label{eq:GF4_companion}
  0\leftrightarrow \begin{pmatrix} 
    0 & 0  \\
    0 & 0
  \end{pmatrix},1 \leftrightarrow\begin{pmatrix} 
    1 & 0  \\
    0 & 1
  \end{pmatrix},\mu\leftrightarrow \begin{pmatrix} 
    0 & 1  \\
    1 & 1
  \end{pmatrix},\mu+1\leftrightarrow \begin{pmatrix} 
  1 & 1  \\
  1 & 0
\end{pmatrix}.
\end{equation} It is easy to check that the addition and
multiplication tables are the same for the two representations. 

\section{Pauli errors as vectors and their commutation
relations via the canonical symplectic form}\label{app:symplectic_form} Let
$\sigma_0\equiv I,\sigma_x\equiv X,\sigma_y\equiv Y,\sigma_z\equiv Z$ denote the
$2\times 2$ Pauli matrices. It is known~\cite{gotte02,Ball_geo} that these
matrices can be mapped to two-bit strings (row vectors) in the following manner
(the vertical lines help readability)
\begin{align}\label{eq:pauli_two_bit_vector} I &\leftrightarrow (0|0),\\\nonumber
X &\leftrightarrow (1|0),\\\nonumber Y &\leftrightarrow (1|1),\\\nonumber Z
&\leftrightarrow (0|1).
\end{align}
The tensor product of Pauli matrices can be obtained by collecting the
 corresponding bit values ($0,1$) from before and after the vertical line of
 each Pauli matrix into two groups, respectively. Then, these two groups are
 concatenated into a new vector~\cite{Ball_geo,gotte02}. They are also separated with a
 vertical line for readability. We illustrate this with the following examples
\begin{align}
  X\otimes Y &\leftrightarrow(11|01),\\\nonumber
  Y\otimes Y &\leftrightarrow(11|11),\\\nonumber
  X\otimes I \otimes I &\leftrightarrow(100|000).
\end{align}

Thus, ignoring phase-factors, there is a bijection between $n$-qubit Pauli
errors and the elements of the vector space $V(2n,2)$ (or in other words
$2n$-bit strings). Moreover, the product of two $n$-qubit Pauli errors is
mapped to addition in the $V(2n,2)$ vector space.

This description is advantageous since we can introduce a
canonical symplectic form that describes when two $n$-qubit Pauli errors
commute using binary addition modulo 2. Let us define this form for two $2n$-bit
strings $u,v\in V(2n,2)$ as 
\begin{equation}\label{eq:symp_form}(u,v)=u\Omega v^T,\end{equation}
where $\Omega$ is a $2n\times 2n$ matrix, composed of the following blocks
\begin{equation}\Omega=\begin{pmatrix} 
  0_{n\times n} & \mathds{1}_{n\times n}  \\
  \mathds{1}_{n\times n} & 0_{n\times n}
\end{pmatrix},\end{equation}
where $0_{n\times n} (\mathds{1}_{n\times n})$ denotes the $n\times n$ all-zero
(identity) matrix.
Two $n$-qubit Pauli errors commute if and
only if their corresponding $u,v\in V(2n,2)$ vectors fulfill that
\begin{equation}(u,v)=0.\end{equation}

It is easy to see that all the requirements hold for the form
in Eq.~\eqref{eq:symp_form} to be symplectic. Namely, it is linear in both
arguments, non-degenerate and alternating. In particular, it is non-degenerate since if
$(u,v)=0$ for $\forall v\in V(2n,2)$ then $u=0,$ which means that only the
identity commutes with all Pauli errors. It is alternating because $(v,v)=0$ for
$\forall v \in V(2n,2),$ implies that every Pauli error commutes with itself.

With this, we can state the following important fact. 
An $s$-dimensional subspace of $V(2n,2)$ spanned by the vectors $(u_1,u_2,...,u_s),$ 
with $u_i\in V(2n,2)$, generates a stabilizer QECC 
if and only if $(u_i,u_j)=0$ for $\forall i,j.$ This means that the corresponding $n$-qubit Pauli
 errors commute.  
In this case the subspace is called totally isotropic with respect to the symplectic form.
Moreover, since we have $s$ independent generators in the $n$-qubit space this 
means that we encode $n-s$ into $n$ qubits.

We emphasize that all symplectic forms are equivalent on $V(2n,2).$ 
This means that if a subspace is totally isotropic with respect to any symplectic form then 
 the corresponding set of Pauli 
errors generates a stabilizer QECC. We use this fact in the proof of Theorem~\ref{thm:points_stab}. 

\section{Constructing the stabilizers from
Theorem~\ref{thm:ptc}}\label{app:example_stabilizers} Here we provide an example
for explicitly constructing the stabilizers constituting the SPTC for the
special case of $r=s=2$ based on Theorem~\ref{thm:ptc}. In this case we need a
$[c,4,d]_4$ LECC. Let us use the code from
Corollary~\ref{cor1}. This means that we have a $[5,4,2]_4$
code with the following generator matrix\begin{equation}\label{eq:gen_example}
  G=\begin{pmatrix} 
    1 & 0 & 0 & 0 & 1 \\
    0 & 1 & 0 & 0 & 1 \\
    0 & 0 & 1 & 0 & 1 \\
    0 & 0 & 0 & 1 & 1 \end{pmatrix}, \end{equation} where 0 (1) is the zero
  (identity) element of $GF(4).$ Each column provides a stabilizer QECC that
  encodes $2(=rs-s)$ qubits into $4(=rs)$ qubits. We take the first column of
  $G$ as a row vector $a=(1,0,0,0)$ over $GF(4)$ and work out explicitly the
  stabilizers. For the remaining columns we only provide the results. Plugging in the
  companion matrix representation for the $GF(4)$ elements that we obtained in
  Eq.~\eqref{eq:companion_GF4} we have that $a$ corresponds to the following $2\times
  8(=s\times 2rs)$ matrix:
  \begin{equation}
    a\equiv\left(
      \begin{array}{cccc|cccc}
      1 & 0 & 0 & 0 & 0 & 0 & 0 & 0\\
      0 & 1 & 0 & 0 & 0 & 0 & 0 & 0\\
      \end{array}
      \right),
  \end{equation}
  where each row represents a generator of the stabilizer. Using the
 correspondence for representing Pauli errors as vectors from
 Eq.~\eqref{eq:pauli_two_bit_vector} we obtain the following stabilizers for the
 first column of $G$:
 \begin{align}\label{eq:stab1}
  X \otimes I\otimes I\otimes I,\\\nonumber
  I \otimes X\otimes I\otimes I.
 \end{align} 
 Similarly, the stabilizers corresponding to the second, third, fourth and fifth
 column of $G$ are as follows
 \begin{align}\label{eq:stab2}
  I\otimes I \otimes X\otimes I,\\\nonumber
  I\otimes I \otimes I\otimes X,
 \end{align}
 \begin{align}\label{eq:stab3}
  Z \otimes I\otimes I\otimes I,\\\nonumber
  I \otimes Z\otimes I\otimes I,
 \end{align}
 \begin{align}\label{eq:stab4}
  I\otimes I \otimes Z\otimes I,\\\nonumber
  I\otimes I \otimes I\otimes Z,
 \end{align} and
 \begin{align}\label{eq:stab5}
  Y\otimes I\otimes Y\otimes I,\\\nonumber
  I\otimes Y\otimes I\otimes Y,
 \end{align}
 respectively. According to Theorem~\ref{thm:ptc} the error probability of the
 SPTC consisting of the stabilizers provided by
 Eqs.~\eqref{eq:stab1},~\eqref{eq:stab2},~\eqref{eq:stab3},~\eqref{eq:stab4},~\eqref{eq:stab5}
 is $\epsilon=1-2/5=3/5.$

\end{document}